\documentclass[12pt]{article}

\usepackage{ifthen}
\newboolean{shortdraft}
\newboolean{qcircuit}

\setboolean{shortdraft}{false} 
\setboolean{qcircuit}{true}     

\ifthenelse{\boolean{shortdraft}}{
\usepackage[paperheight=8in, paperwidth=6.5in, margin=0.2in]{geometry}
}{}

\usepackage{amsmath}
\usepackage{amsfonts}
\usepackage{amsthm}
\usepackage{amssymb}
\usepackage{url}
\usepackage{hyperref}

\ifthenelse{\boolean{qcircuit}}{
	\input{Qcircuit}
}{
\newcommand{\ket}[1]{\left | \, #1 \right \rangle}
\newcommand{\bra}[1]{\left \langle #1 \, \right |}
}

\ifthenelse{\boolean{shortdraft}}{
	\renewcommand{\cite}[1]{[??]}
}{
	\bibliographystyle{halphads}
}

\newcommand{\norm}[1]{\left|\left|#1\right|\right|}

\newtheorem{definition}{Definition}

\newtheorem{theorem}{Theorem}

\newtheorem{lemma}{Lemma}

\title{Self-testing graph states}
\author{Matthew McKague \\ Centre for Quantum Technologies \\ National University of Singapore \\Ê\url{matthew.mckague@nus.edu.sg}}

\begin{document}
\maketitle
\begin{abstract}
We give a construction for a self-test for any connected graph state.  In other words, for each connected graph state we give a set of non-local correlations that can only be achieved (quantumly) by that particular graph state and certain local measurements.  The number of correlations considered is small, being linear in the number of vertices in the graph.  We also prove robustness for the test.

\end{abstract}

\section{Introduction}
Self-testing is a process where a skeptical classical user attempts to verify the operation of a collection of quantum devices without trusting any of them \emph{a priori}.  Importantly, we wish to make as few assumptions as possible about the operation of the devices and in particular we do not bound the dimension of the state space for each device.  However we do make the necessary assumption that the quantum devices are not allowed to communicate with each other.  Despite these severe restrictions on our knowledge it is possible to devise self-tests for a number of different situations.

Self-testing was first introduced by Mayers and Yao \cite{Mayers:2004:Self-testing-qu} who described a self-test for a maximally entangled pair of qubits (EPR pair) along with a small set of local measurements.  Meanwhile, self-testing of gates was introduced by van Dam et al. \cite{van-Dam:1999:Self-Testing-of} in the scenario of known Hilbert space dimensions.  These two results were extended to testing of circuits over a real Hilbert space by Magniez et al \cite{Magniez:2006:Self-testing-of}.  Most recently, McKague and Mosca \cite{McKague:2010:Generalized-sel} reproved the Mayers-Yao result and extended it to allow for testing of a larger set of measurements including measurements over the full complex Hilbert space.

In this paper we use proof techniques developed in \cite{McKague:2010:Generalized-sel} to define self-tests for the graph state for any connected graph.  This family of self-tests is efficient in the number of measurement settings, requiring only two or three measurement settings on each vertex, depending on the graph.  As well the total number of correlations tested is small, only one per vertex plus an additional 3 at most.  We also prove that the self-tests are robust.

\subsection{Graph states and notation}
A graph $G$ is composed of two sets:  a set $V$ of \emph{vertices}, and a set $E \subset V \times V$ of \emph{edges}.  For our purposes we suppose that $(v,v) \notin E$ and $(v,u) \in E$ whenever $(u,v) \in E$.  Two vertices $u,v$ are said to be \emph{adjacent} if $(u,v) \in E$.  A cycle is a sequence of vertices in which each vertex occurs at most once, each vertex in the sequence is adjacent to the next vertex in the sequence, and the last vertex is adjacent to the first.  A \emph{subgraph} $G^{\prime}$ of $G$ is a graph $(E^{\prime}, V^{\prime})$ with $E^{\prime} \subseteq E$, $V^{\prime} \subseteq V$.  An \emph{induced subgraph} is a subgraph in which $E^{\prime} = \{(u,v) \in E | u,v \in V^{\prime}\}$, so the subgraph contains all edges between vertices of $V^{\prime}$ in the original graph.  The \emph{neighbours} $N_{v}$ of a vertex $v$ are the vertices to which $v$ is connected with an edge, i.e. $N_{v} = \{ u \in V | (u,v) \in E\}$.  A \emph{bipartite} graph is a graph in which the set of vertices may be partitioned into two sets $S$ and $T$, each of which has no edges within it.  So the induced subgraphs on $S$ and $T$ have no edges.  An important property of bipartite graphs is that they are exactly the graphs which contain no cycles with an odd number of vertices.  A graph  is \emph{connected} if for each pair of vertices $u,v$ there is a sequence of adjacent vertices beginning with $u$ and ending in $v$.  For more detail regarding graph theory see Diestel \cite{Diestel:2010:Graph-Theory}.

A graph state consists of a set of qubits indexed by the set of vertices $V$, each prepared in the state $\ket{+}_{v} = \frac{1}{\sqrt{2}}\left(\ket{0}_{v} + \ket{1}_{v}\right)$, followed by $(\text{CTRL}-Z)_{uv}$ operations between pairs of qubits where the corresponding vertices $u,v$ in the graph are adjacent.  If the graph is not connected then the graph state will be a product state of graph states on the separate components.  Hence connected graphs form the interesting cases.

Graph states are also characterized by their stabilizer group.  Let the operators $X_{v}$ and $Z_{v}$ be the Pauli operators $X$ and $Z$ applied to qubit $v$, tensor product with $I$ on all other qubits.  If $P$ is a Pauli and $S \subseteq V$ then 
\begin{equation}
P^{S} = \prod_{v \in S} P_{v}.
\end{equation}
The stabilizer group for a graph state on the graph $G = (V, E)$ is generated by
\begin{equation} 
S_{v} = \left\{X_{v} Z^{N_{v}} | v \in V \right\}.
\end{equation}
That is, for each vertex $v$ there is a stabilizer operator with $X$ operating on $v$ and $Z$ operating on each of $v$'s neighbours.  Note that there are $n$ such operators, they pairwise commute and are independent.  Hence there is exactly one state with this stabilizer group.  That is to say, the graph state $\ket{\psi}$ is the unique state for which $S_{v}\ket{\psi} = \ket{\psi}$ for each $v \in V$.

As one additional piece of notation, we will frequently need to deal with products of stabilizers on a subset of vertices.  For this case we define
\begin{equation}
Z^{N(S)} = \prod_{v \in S} Z^{N_{v}}
\end{equation}
where the factor $Z_{v}$ appears if $v$ has an \emph{odd} number of neighbours in $S$.

\subsection{Self-testing definitions}
Consider the following \emph{black-box} scenario:  we are given a set of devices, each with a knob labeled with a number of settings, a pair of lights labeled $\pm 1$, and a button.  After we select a setting and push the button one of the lights turns on.  We are told that the devices jointly share a state which is measured, according to the knob setting, in a specified basis.  Our goal is to determine if the black-boxes are operating according to their specification using only the external controls of the boxes.  Additionally we may isolate the boxes to ensure that they do not communicate.

We begin with a \emph{reference experiment} consisting of an $n$-partite system in the state $\ket{\psi}$ together with local measurement observables $M_{j,m}$ on subsystem $j$ with measurement setting $m \in \{0,1, \dots, k_{j}\}$.  The measurement setting $m=0$ corresponds to no measurement, which we may represent with the identity.  The reference experiment represents the specification for how the black-boxes supposedly operate.  In particular, we assume that the state and observables are known.

In addition, we have a \emph{physical experiment} consisting of an $n$-partite physical system in the 
state\footnote{We consider only pure states, but since the Hilbert space of the physical system has unbounded dimension we may easily add a purification to mixed states.} 
$\ket{\psi^{\prime}}$ together with local measurement observables $M^{\prime}_{j,m}$ on subsystem $j$, with $m \in \{0,1, \dots, k_{j}\}$.  Again we may take $M^{\prime}_{j,0} = I$ indicating that we do not measure the subsystem.  We place no bound on the dimension of the Hilbert space of each subsystem, but assume that it is finite.  The physical experiment represents how the black-boxes \emph{actually} operate.

If a physical and reference experiment have the same number of subsystems and the same number of measurements on each subsystem, then we say that they are \emph{compatible}.  Note that we will always deal with the case of two-outcome measurements, so that all observables have eigenvalues $\pm 1$.  In principle, though, the definitions can be extended to other types of measurements.

To be more specific about our task, we introduce two notions, \emph{simulation} and \emph{equivalence}.

\begin{definition}\label{def:simulates}
Let a physical experiment and a compatible reference experiment be given as above.   We say that the physical experiment \emph{simulates} the reference experiment if for each measurement setting $m = (m_{1}, \dots, m_{n})$, $m_{j} \in \{0, \dots, k_{j}\}$ we have
\begin{equation}
\bra{\psi^{\prime}} \bigotimes_{j=1}^{n} M^{\prime}_{j, m_{j}} \ket{\psi^{\prime}} = \bra{\psi} \bigotimes_{j=1}^{n} M_{j, m_{j}} \ket{\psi}.
\end{equation}
\end{definition}

Here it will be sufficient to consider only a subset of possible measurement settings.  In this case we include the measurement settings of interest in our description of the reference experiment.

\begin{definition}\label{def:equivalence}
Let a physical experiment and a compatible reference experiment be given as above.  We say that the physical experiment is \emph{equivalent} to the reference experiment if there exists a local isometry
\begin{equation}
\Phi = \Phi_{1} \otimes \dots  \otimes \Phi_{n}
\end{equation}
and a state $\ket{junk}$ such that, for each $j$, and $m \in \{1, \dots, k_{j}\}$
\begin{eqnarray}
\Phi(\ket{\psi^{\prime}}) & = & \ket{junk}\otimes \ket{\psi}\\
\Phi(M^{\prime}_{j,m}\ket{\psi^{\prime}}) & = & \ket{junk} \otimes M_{j,m}\ket{\psi}
\end{eqnarray}
where $\ket{junk}$ is in the same Hilbert space as $\ket{\psi^{\prime}}$.
\end{definition}

When describing any physical system we must first fix a reference frame, and decide which components to describe and which to leave out.  Thus we may take a description and apply local changes of basis, or add ancillas and arrive at another, perfectly acceptable, description of the system.  These two operations are invisible from the perspective of classical interactions with devices so we can never rule them out.  This motivates our definition of equivalence, which takes such ambiguities in quantum descriptions into account.

Throughout the remainder of this paper we will used primed ($\ket{\psi^{\prime}}$, $X^{\prime}$, $S^{\prime}_{v}$ etc.) to denote physical measurements and states and unprimed for reference measurements and states.  Note that $S^{\prime}_{v} = X^{\prime}_{v} \otimes Z^{\prime N(v)}$ and other derived physical measurements are defined in terms of the local physical measurements.  Also, although we use the letters $X$ and $Z$ for the physical measurements, these need not be Pauli matrices, and we assume nothing about their structure other than what we mention explicitly.

\subsection{Main results}

A self-testing theorem specifies a particular reference experiment and states that if a physical experiment simulates the reference experiment, then it is equivalent to it.  That is to say, for a particular experiment \emph{simulation implies equivalence.}  Our main result is to show that this is the case for the following two reference experiments.

\begin{definition}[Reference experiment 1: connected graph with an odd induced cycle]
Let $G= (V,E)$ be a connected graph containing an odd induced cycle $C= (V^{\prime}, E^{\prime})$.  Let $\ket{\psi}$ be the corresponding graph state with stabilizers $S_{v}$,  $v \in V$.  The reference experiment consists of the state $\ket{\psi}$, the stabilizer measurements $S_{v}$ and the measurement $X^{V^{\prime}} Z^{N(V^{\prime})}$.
\end{definition}

It is easy to show that a graph which contains any odd cycle contains an induced cycle.  Thus reference experiment 1 is applicable to all connected non-bipartite graphs.

\begin{definition}[Reference experiment 2: connected graph]
Let $G=(V,E)$ be a connected graph with at least two vertices.  Let $\ket{\psi}$ be the corresponding graph state with stabilizers $S_{v}, v \in V$.  Choose a fixed edge $(u,v) \in E$ and define
\begin{equation}
D_{u} = \frac{1}{\sqrt{2}}\left(X_{u} + Z_{u}\right)
\end{equation}
The reference experiment consists of the state $\ket{\psi}$, the stabilizer measurements $S_{v}$ and the measurements
\begin{equation}
Z^{\prime}_{u} Z^{\prime N_{u}}
\end{equation}
\begin{equation}
D_{u} Z^{N_{u}}
\end{equation}
\begin{equation}
D_{u} X_{v} Z^{N_{v} \setminus \{u\}}
\end{equation}

\end{definition}

In appendix~\ref{sec:bipartiteneedsd} we show that the $D$ measurements are required since for a bipartite graph all measurements using $X$ and $Z$ alone can be simulated using a classical hidden variable model.

\begin{theorem}\label{theorem:graphtesting}
If a physical experiment is compatible with reference experiment 1 (2), and simulates it, then the physical experiment is equivalent to reference experiment 1(2).
\end{theorem}
%

\section{Proof of main result}
The proof consists of three sections.  First we determine the expected values for the measurements in the reference experiment.  Next we show that if the physical experiment simulates the reference experiment then the $X^{\prime}$ and $Z^{\prime}$ operators anti-commute.  Finally we construct the local isometry and use the anti-commuting property of the $X^{\prime}$ and $Z^{\prime}$ operators to show equivalence.

\subsection{Probability distribution from graph states}
We first derive the probability distributions that arise from a graph state with trusted measurements.  This establishes the conditions that a physical experiment must meet in order to simulate the reference experiment.

Clearly, the stabilizer measurements all satisfy
\begin{equation}
\bra{\psi}S_{v} \ket{\psi} = 1.
\end{equation}
For reference experiment 1, we need one additional measurement.
\begin{lemma}\label{lemma:evendegreeoddedges}
Let $G= (V,E)$ be a graph and let $\ket{\psi}$ be the corresponding graph state.  Let $V^{\prime} \subseteq V$ and let $G^{\prime} = (V^{\prime}, E^{\prime})$ be the induced subgraph on $V^{\prime}$.  If each $v \in V^{\prime}$ has even degree then
\begin{equation}
(-1)^{|E^{\prime}|}X^{V^{\prime}} Z^{N(V^{\prime})} \ket{\psi} = \ket{\psi}
\end{equation}
\end{lemma}
\begin{proof}
Consider the product
\begin{equation}\label{eq:subgraphstabilizers}
\left(\prod_{v \in V^{\prime}} S_{v}\right) \ket{\psi}
\end{equation}
First note that there will be an $X_{v}$ factor for each $v \in V^{\prime}$.  As well, there will be a $Z_{u}$ factor for each $v \in V^{\prime}$ adjacent to $u$.  Canceling pairs we see that there will be an overall $Z_{u}$ factor exactly when there are an odd number of neighbours of $u$ in $V^{\prime}$.  Hence the $Z$ factor will be $Z^{N(V^{\prime})}$.  We only need to determine the sign.  Note that the $Z_{u}$, $u \notin V^{\prime}$ factor all commute so we need not consider them any more.

The order of multiplication in equation~(\ref{eq:subgraphstabilizers}) does not matter since the stabilizers all commute.  For convenience, then, we may write the product as the product of the rows of a matrix with each column corresponding to a $v \in V^{\prime}$ and each row a stabilizer.  We choose the order of the rows so that the $X$s appear along the diagonal\footnote{The matrix may be constructed by taking the adjacency matrix of $G^{\prime}$, which has a 1 in the $u,v$ position when $(u,v) \in E^{\prime}$, replacing the diagonal with $X$s, the 0s with $I$s and the $1$s with $Z$.}.  For a 5-cycle, for instance, we have
\begin{equation}
\begin{matrix}
X & Z & I & I & Z \\
Z & X & Z & I & I \\
I & Z & X & Z & I \\
I & I & Z & X & Z \\
Z & I & I & Z & X \\
\end{matrix}. 
\end{equation}
The factor on each vertex equals the product of the entries in the corresponding column.  In each column there is one $X$ and one $Z$ for each neighbour.  The factor will be either $\pm XZ$ or $\pm X$, depending on whether there is an odd or even number of $Z$s.  The sign depends on the number of $Z$s above the $X$, since we must use the fact that $XZ = - ZX$ once for each such $Z$.  Combining the signs from all vertices, there is a $-1$ factor for each $Z$ above the diagonal, and hence one for each edge in $G^{\prime}$.  The overall sign, then, is $(-1)^{|E^{\prime}|}$.
\end{proof}

For reference experiment 1 we consider an odd cycle, and hence we obtain
\begin{equation}
\bra{\psi}X^{V^{\prime}} Z^{N(V)^{\prime}} \ket{\psi} = -1.
\end{equation}

Reference experiment 2 has three measurements other than the stabilizer.  First we have $Z_{u}Z^{N_{u}}$.  This is just $S_{u}$ with $X_{u}$ replaced by $Z_{u}$.  Since $X$ and $Z$ anti-commute we have
\begin{equation}
\bra{\psi}Z_{u}Z^{N_{u}}\ket{\psi} = 0.
\end{equation}
From this, and linearity, we obtain
\begin{equation}
\bra{\psi}D_{u} Z^{N_{u}} \ket{\psi} = \frac{1}{\sqrt{2}}.
\end{equation}
Finally, the operator $D_{u} X_{v} Z^{N_{v} \setminus \{u\} }$ is a linear combination of $S_{v}$ and $S_{v}$ with $Z_{u}$ replaced with $X_{u}$.  As above, then, we find
\begin{equation}
\bra{\psi}D_{u} X_{v} Z^{N_{v} \setminus \{u\} }\ket{\psi} = \frac{1}{\sqrt{2}}.
\end{equation}
%

\subsection{Statistics imply anti-commuting observables}
We now suppose that the physical experiment simulates either reference experiment 1 or 2 and show that this implies that the $X^{\prime}$ and $Z^{\prime}$ measurements on each vertex anti-commute (on the support of $\ket{\psi}$).

First, note that $\bra{\psi^{\prime}}S^{\prime}_{v}\ket{\psi^{\prime}} = 1$ implies $S^{\prime}_{v} \ket{\psi^{\prime}} = \ket{\psi^{\prime}}$, and similarly for other measurements.  This allows us to immediately drop probabilities and deal with states directly.

As a first step towards our goal, we prove a type of induction lemma which says that if the $X^{\prime}$ and $Z^{\prime}$ observables anti-commute on vertex, then the same is true for an adjacent vertex.  Thus we need only show anti-commuting observables on one vertex, and apply the lemma repeatedly along paths to all other vertices (since $G$ is connected.)  

\begin{lemma}\label{lemma:chaining}
Given a graph $G$ with $(u,v) \in E$.  If observables $X^{\prime}_{v}, Z^{\prime}_{v}, X^{\prime}_{u}, Z^{\prime}_{u}$, and $ \{Z^{\prime}_{w}| w \in N_{u} \cup N_{v}\}$ and state $\ket{\psi^{\prime}}$ satisfy
\begin{equation}
S^{\prime}_{u} \ket{\psi^{\prime}} = S^{\prime}_{v} \ket{\psi^{\prime}} = \ket{\psi^{\prime}}
\end{equation}
\begin{equation}
(X^{\prime}Z^{\prime})_{v} \ket{\psi^{\prime}} = - (Z^{\prime}X^{\prime})_{v} \ket{\psi^{\prime}}
\end{equation}
then %
\begin{equation}
(X^{\prime}Z^{\prime})_{u} \ket{\psi^{\prime}} = - (Z^{\prime}X^{\prime})_{u} \ket{\psi^{\prime}}
\end{equation}

\end{lemma}
\begin{proof}
From the fact that $(u,v) \in E$ we obtain
\begin{eqnarray}
(Z^{\prime}X^{\prime})_{u}\ket{\psi^{\prime}} & = & (Z^{\prime}X^{\prime})_{u}S^{\prime}_{u} S^{\prime}_{v}S^{\prime}_{u}S^{\prime}_{v} \ket{\psi^{\prime}} \\
 & = & (Z^{\prime}X^{\prime})_{u}X^{\prime}_{u} Z^{\prime}_{v} X^{\prime}_{v} Z^{\prime}_{u} X^{\prime}_{u} Z^{\prime}_{v} X^{\prime}_{v} Z^{\prime}_{u} \ket{\psi^{\prime}} \\
& = & (X^{\prime}Z^{\prime})_{u} (Z^{\prime}X^{\prime})_{v} (Z^{\prime}X^{\prime})_{v}\ket{\psi^{\prime}} \\
& = & -(X^{\prime}Z^{\prime})_{u}(Z^{\prime}X^{\prime})_{v} (X^{\prime}Z^{\prime})_{v}\ket{\psi^{\prime}} \\
& = & - (X^{\prime}Z^{\prime})_{u} \ket{\psi^{\prime}}\\
\end{eqnarray}
\end{proof}

For reference experiment 1 we show that the observables $X^{\prime}$ and $Z^{\prime}$ anti-commute for each vertex in the induced odd cycle. 

\begin{lemma}\label{lemma:oddcycleanticommute} 
Let $G= (E,V)$ be a connected graph and let $C = (E^{\prime}, V^{\prime})$ be an induced odd cycle of $G$ and let $u \in V^{\prime}$.  If observables $X^{\prime}_{u}, Z^{\prime}_{u}$ for $u \in V^{\prime}$,  $\{Z^{\prime}_{w}| w \text{ has a neighbour in } C \}$ and state $\ket{\psi^{\prime}}$ satisfy
\begin{equation}
S^{\prime}_{u} \ket{\psi^{\prime}} = \ket{\psi^{\prime}}
\end{equation}
\begin{equation}
-X^{^{\prime}V^{\prime}} Z^{\prime N(V^{\prime})} \ket{\psi^{\prime}} = \ket{\psi^{\prime}}
\end{equation}

 Then $(X^{\prime}Z^{\prime})_{u}\ket{\psi} = - (Z^{\prime}X^{\prime})_{u}\ket{\psi}$ for each $u \in V^{\prime}$.
\end{lemma}
\begin{proof}
Number the vertices in the cycle $1$ through $k$ so $1$ is adjacent to $2$, etc..  Without loss of generality we may assume that $u$ is vertex $1$.  We next consider the following state:
\begin{equation}
-X^{\prime V^{\prime}} Z^{\prime N(V^{\prime})} \prod_{j=1}^{\frac{k-1}{2}} S^{\prime}_{2j} \prod_{j=1}^{\frac{k-1}{2}} S^{\prime}_{2j-1} \ket{\psi^{\prime}} = \ket{\psi^{\prime}}
\end{equation}
Note that the factor $Z^{\prime N(V^{\prime})}$ is cancelled by $Z$ operations arising from the products of the $S^{\prime}_{v}$.  We may write the product as the product of the rows of the following matrix, where column $j$ corresponds to vertex $j$ in the cycle:
\begin{equation}
	\begin{matrix}
	-X^{\prime} & X^{\prime} & X^{\prime} & X^{\prime} & X^{\prime} &  \dots &X^{\prime} & X^{\prime} & X^{\prime} \\
	Z^{\prime} & X^{\prime} & Z^{\prime} & I & I &  \dots & I & I & I \\
	I & I & Z^{\prime} & X^{\prime} & Z^{\prime} &  \dots & I & I & I \\
	& & & & & \vdots  \\
	I & I & I & I & I & \dots & Z^{\prime} & X^{\prime} & Z^{\prime} \\
	X^{\prime} & Z^{\prime} & I & I & I & \dots & I & I & Z^{\prime} \\
	I & Z^{\prime} & X^{\prime} & Z^{\prime} & I & \dots & I & I & I \\
	I & I & I & Z^{\prime} & X^{\prime} & \dots & I & I & I \\
			& & & & & \vdots  \\
	Z^{\prime} & I & I & I & I & \dots & I & Z^{\prime} & X^{\prime} \\		
	\end{matrix}
\end{equation}
In each column there are two $X^{\prime}$ operators and two $Z^{\prime}$ operators.  Also, their arrangement is such that, for every column except the first, the two $X^{\prime}$ operators are next to one another, so they cancel directly, and similarly for the $Z^{\prime}$ operators.  Hence
\begin{equation}
-(X^{\prime}Z^{\prime})_{u} (X^{\prime}Z^{\prime})_{u} \ket{\psi^{\prime}} = \ket{\psi^{\prime}}
\end{equation}
The desired result follows immediately.
\end{proof}

For reference experiment 2, we have one additional measurement on a particular vertex $u$.  We use this extra measurement to establish that the $X^{\prime}$ and $Z^{\prime}$ measurements on $u$ anti-commute.

\begin{lemma}\label{lemma:dxyanticommute}
  Let $G= (V,E)$ be a connected graph with $(u,v) \in E$.  If observables $D^{\prime}_{u}, X^{\prime}_{v}, Z^{\prime}_{v}, X^{\prime}_{u}, Z^{\prime}_{u}, \{Z^{\prime}_{w}| w \in N_{u} \cup N_{v}\}$ and state $\ket{\psi^{\prime}}$ satisfy
\begin{equation}
S^{\prime}_{u} \ket{\psi^{\prime}} = S^{\prime}_{v} \ket{\psi^{\prime}} = \ket{\psi^{\prime}}
\end{equation}
\begin{eqnarray}
\bra{\psi^{\prime}} Z^{\prime}_{u} Z^{\prime N_{u}} \ket{\psi^{\prime}} & = & 0 \\ 
\bra{\psi^{\prime}} D^{\prime}_{u} Z^{\prime N_{u}} \ket{\psi^{\prime}} & = &  \frac{1}{\sqrt{2}} \\
\bra{\psi^{\prime}} D^{\prime}_{u} X_{v} Z^{\prime N_{v} \setminus u} \ket{\psi^{\prime}} & = & \frac{1}{\sqrt{2}} \\
\end{eqnarray}
then 
$-(X^{\prime}Z^{\prime})_{u} \ket{\psi^{\prime}} = (Z^{\prime}X^{\prime})_{u}\ket{\psi^{\prime}}$
\end{lemma}
\begin{proof}
Since $\bra{\psi^{\prime}} X^{\prime}_{u} Z^{\prime N_{u}} \ket{\psi} = 1$ we have $X^{\prime}_{u} \ket{\psi^{\prime}} = Z^{\prime N_{u}} \ket{\psi}$.  Similarly, $Z^{\prime}_{u} \ket{\psi^{\prime}} = X^{\prime}_{v} Z^{\prime N_{v} \setminus u}\ket{\psi^{\prime}}$.  Along with $\bra{\psi^{\prime}} Z^{\prime}_{u} Z^{\prime N_{u}} \ket{\psi^{\prime}} = 0 $ we find that $X^{\prime}_{u} \ket{\psi^{\prime}}$ is orthogonal to $Z^{\prime}_{u} \ket{\psi^{\prime}}$.  We also obtain
$\bra{\psi^{\prime}} D^{\prime}_{u} Z^{\prime}_{u}\ket{\psi^{\prime}} =  \frac{1}{\sqrt{2}}$ and  $\bra{\psi^{\prime}} D^{\prime}_{u} X^{\prime}_{u}\ket{\psi^{\prime}} =  \frac{1}{\sqrt{2}}$.  Since $D^{\prime}_{u} \ket{\psi^{\prime}}$ has norm 1, we find
\begin{equation}
D^{\prime}_{u} \ket{\psi^{\prime}} = \frac{1}{\sqrt{2}} X^{\prime}_{u}\ket{\psi^{\prime}} + Z^{\prime}_{u}\frac{1}{\sqrt{2}} \ket{\psi^{\prime}}
\end{equation}

Further, since $\left(D^{\prime}_{u}\right)^{2} = I = \left(Z^{\prime}_{u}\right)^{2} = \left(X^{\prime}_{u}\right)^{2}$, and 
\begin{eqnarray}
\ket{\psi^{\prime}} & = & \left(D^{\prime}_{u}\right)^{2} \ket{\psi^{\prime}} \\ 
 &=  &\frac{1}{\sqrt{2}} D^{\prime}_{u} \left(Z^{\prime N_{u}} + X_{v} Z^{\prime N_{v} \setminus u}\right) \ket{\psi^{\prime}} \\
&  = & \frac{1}{2}\left(Z^{\prime N_{u}} + X_{v} Z^{\prime N_{v} \setminus u}\right)\left( X^{\prime}_{u} + Z^{\prime}_{u}\right) \ket{\psi^{\prime}} \\
& = & \frac{1}{2}\left( 2I + (X^{\prime}Z^{\prime})_{u} + (Z^{\prime} X^{\prime})_{u}\right) \ket{\psi^{\prime}} \\
\end{eqnarray}
In order for this to be true, we must have
\begin{equation}
(X^{\prime}Z^{\prime})_{u} \ket{\psi^{\prime}} = - (Z^{\prime} X^{\prime})_{u}\ket{\psi^{\prime}}.
\end{equation}

\end{proof}

We conclude with a technical lemma that allows us to exchange $X^{\prime}_{v}$ operations for $Z^{\prime}_{v}$ operations.

\begin{lemma}\label{lemma:exchangexy}
Let $G=(V,E)$ be a connected graph and let $X^{\prime}_{v}, Z^{\prime}_{v}$ for $vÊ\in V$ and $\ket{\psi^{\prime}}$ (and $D_{u}$ for some $u\in V$) be a physical experiment that simulates reference test 1 (or 2).  Let $G^{\prime}= (V^{\prime}, E^{\prime})$ be an induced subgraph of $G$.  Then
\begin{equation}
(-1)^{|E^{\prime}|} X^{\prime V^{\prime}} \ket{\psi^{\prime}} = Z^{\prime N(V^{\prime})} \ket{\psi^{\prime}}
\end{equation}
\end{lemma}
\begin{proof}
We use the previous lemmas to conclude that $X^{\prime}_{v} Z^{\prime}_{v}\ket{\psi^{\prime}} = - Z^{\prime}_{v} X^{\prime}_{v} \ket{\psi^{\prime}}$ for each $v$.  Then we repeat the argument used in the proof of lemma~\ref{lemma:evendegreeoddedges}.  Essentially, we look at the product
\begin{equation}
\prod_{v} S^{\prime}_{v} \ket{\psi^{\prime}}.
\end{equation}
Writing this product out as a the product of rows of a symmetric matrix with $X^{\prime}$s along the diagonal, we see that in order to get all the $X^{\prime}$s together we must use the anti-commuting relation once for each $Z^{\prime}$ above the diagonal.  Since there is one $Z^{\prime}$ above the diagonal for each edge, we obtain the factor $(-1)^{|E^{\prime}|}$.

\end{proof}
%

\subsection{Constructing the isometry}
The local isometry $\Phi$ that we use to show equivalence between the physical experiment and the reference experiment is the tensor product of isometries $\Phi_{v}$ for various $v \in V$, is in the circuit shown in figure~\ref{fig:epr_local_unitary_circuit}.

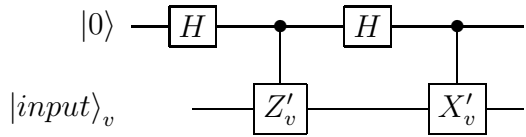
\begin{figure}[h]
\[
\Qcircuit @C=0.5cm @R=0.5cm{
\lstick{\ket{0}}  & \gate{H} & \ctrl{1}  & \gate{H} & \ctrl{1} &\qw \\
\lstick{\ket{input}_{v}} &  & \gate{Z^{\prime}_{v}} & \qw & \gate{X^{\prime}_{v}} & \qw\\
}
\]
\caption{Circuit for $\Phi_{v}$}
\label{fig:epr_local_unitary_circuit}
\end{figure}

The circuit is based on the argument used by Mayers and Yao in their original EPR test.  It may be seen as a type of SWAP gate, decomposed into three CNOT gates.  Here the first CNOT gate is omitted since the ancilla is always initialized in the state $\ket{0}$.  The Hadamards and Controlled $Z$ operation replace a CNOT targeted on the ancilla.  With these points in mind, we see that when $Z^{\prime}_{v}$ and $X^{\prime}_{v}$ are indeed qubit Pauli operators the circuit defines a SWAP operation.

We will now calculate the result of $\Phi$ applied to $\ket{\psi^{\prime}}$.
\begin{equation}
\Phi(\ket{\psi^{\prime}}) = \frac{1}{2^{n}}\sum_{x} \bigotimes_{vÊ\in V} X^{\prime x_{v}}_{v}\left(I + (-1)^{x_{v}}Z^{\prime}_{v}\right) \ket{\psi^{\prime}}\ket{x}
\end{equation}
with $x= (x_{v})_{v \in V} \in \{0,1\}^{|V|}$.  Applying the anti-commutation relation, this simplifies to 
\begin{equation}\label{eq:isometryanticommute}
\Phi(\ket{\psi^{\prime}}) = \frac{1}{2^{n}}\sum_{x} \bigotimes_{vÊ\in V} \left(I + Z^{\prime}_{v}\right)X^{\prime x_{v}}_{v} \ket{\psi^{\prime}}\ket{x}.
\end{equation}
Using lemma~\ref{lemma:exchangexy} and the fact that $(I + Z^{\prime}_{v}) Z^{\prime}_{v} = I + Z^{\prime}_{v}$ we finally find
\begin{equation}\label{eq:isometryfinal}
\Phi(\ket{\psi^{\prime}}) = \left(\frac{1}{\sqrt{2^{n}}}\bigotimes_{vÊ\in V} \left(I + Z^{\prime}_{v}\ket{\psi^{\prime}}\right)\right) \left(\frac{1}{\sqrt{2^{n}}} \sum_{x} (-1)^{e(x)}\ket{x}\right)
\end{equation}
where $e(x)$ is the number of edges in the induced subgraph on the set $V_{x} = \{v \in V | x_{v} = 1\}$.  

Set $\ket{\phi} = \frac{1}{\sqrt{2^{n}}} \sum_{x} (-1)^{e(x)} \ket{x}$.  Consider $S_{v} \ket{x}$ for some $x$.  This will be $\pm \ket{x \oplus 1_{v}}$ where $1_{v}$ is the binary vector with 1 in position $v$ and 0 everywhere else.  The sign may be computed as follows:  for each $Z_{u}$ component of $S_{v}$, if $x_{u} = 1$ a $-1$ factor will be introduced.  This happens when $(u,v) \in E$ and $u$ is in $V_{x}$.  We may see this as either removing or adding the vertex $v$ and adding a $-1$ factor for each edge between $v$ and another vertex in $V_{x}$.  Thus $S_{v} (-1)^{e(x)} \ket{x} = (-1)^{e(x \oplus 1_{v})} \ket{x \oplus 1_{v}}$.  In other words, this exactly produces the correct sign on each $\ket{x}$ so that $S_{v} \ket{\phi} = \ket{\phi}$ and in fact $\ket{\phi}Ê= \ket{\psi}$.

Now consider $\Phi(X^{\prime}_{v} \ket{\psi^{\prime}})$ for some $v$.  After anti-commuting the $X^{\prime}$ operations we have
\begin{equation}\label{eq:isometryanticommute2}
\Phi(X^{\prime}_{u}\ket{\psi^{\prime}}) = \frac{1}{2^{n}}\sum_{x} \bigotimes_{vÊ\in V} \left(I + Z^{\prime}_{v}\right)X^{\prime x_{v}}_{v} X^{^{\prime}}_{u}\ket{\psi^{\prime}}\ket{x}.
\end{equation}
In this equation, we may simply replace $X^{\prime x_{v}}_{v} X^{^{\prime}}_{u}$ with $X^{\prime x_{v} \oplus 1_{u}}_{v}$, where $1_{u}$ is the vector with 0s everywhere, except position $u$.  After applying lemma~\ref{lemma:exchangexy} we arrive at
\begin{equation}\label{eq:isometryfinal2}
\Phi(X^{\prime}_{u}\ket{\psi^{\prime}}) = \left(\frac{1}{2^{n}}\bigotimes_{vÊ\in V} \left(I + Z^{\prime}_{v}\right)\ket{\psi^{\prime}}\right) \sum_{x} (-1)^{e(x \oplus 1_{u})}\ket{x}.
\end{equation}
A change of variable, $x \mapsto x \oplus 1_{u}$, and the fact that $X_{u}\ket{x} = \ket{x \oplus 1_{u}}$ gives the final result, 
\begin{equation}
\Phi(X^{\prime}_{v}\ket{\psi^{\prime}}) = \left(\frac{1}{\sqrt{2^{n}}}\bigotimes_{vÊ\in V} \left(I + Z^{\prime}_{v}\right)\ket{\psi^{\prime}}\right) X_{v}\ket{\psi}.
\end{equation}

A similar analysis shows that
\begin{equation}
\Phi(Z^{\prime}_{v}\ket{\psi^{\prime}}) = \left(\frac{1}{\sqrt{2^{n}}}\bigotimes_{vÊ\in V} \left(I + Z^{\prime}_{v}\right)\ket{\psi^{\prime}}\right) Z_{v}\ket{\psi}.
\end{equation}
Recall from the proof of lemma~\ref{lemma:dxyanticommute} that $D^{\prime}_{v}\ket{\psi^{\prime}}$ may be written as $D^{\prime}_{v}\ket{\psi^{\prime}} = \frac{1}{\sqrt{2}}\left(X^{\prime}_{v} + Z^{\prime}_{v}\right)\ket{\psi^{\prime}}$.  By linearity, then
\begin{equation}
\Phi(D^{\prime}_{v}\ket{\psi^{\prime}}) = \left(\frac{1}{\sqrt{2^{n}}}\bigotimes_{vÊ\in V} \left(I + Z^{\prime}_{v}\right)\ket{\psi^{\prime}}\right) D_{v}\ket{\psi}.
\end{equation}

This concludes the proof of theorem~\ref{theorem:graphtesting}.

\section{Robustness}

In this section we will show that the main theorems are both robust.

\subsection{Definitions and main theorem}
First, we modify the definitions of simulation and equivalence to allow for small deviations from the reference experiment

\begin{definition}\label{def:epsilonsimulates}
Let a physical experiment and a compatible reference experiment be given as above.   We say that the physical experiment \emph{$\epsilon$-simulates} the reference experiment if for each measurement setting $m = (m_{1}, \dots, m_{n})$, $m_{j} \in \{0, \dots, k_{j}\}$ we have
\begin{equation}
\left|\bra{\psi^{\prime}} \bigotimes_{j=1}^{n} M^{\prime}_{j, m_{j}} \ket{\psi^{\prime}} - \bra{\psi} \bigotimes_{j=1}^{n} M_{j, m_{j}} \ket{\psi} \right| \leq \epsilon.
\end{equation}
\end{definition}

\begin{definition}\label{def:deltaequivalence}
Let a physical experiment and a compatible reference experiment be given as above.  We say that the physical experiment is \emph{$\delta$-equivalent} to the reference experiment if there exists a local isometry
\begin{equation}
\Phi = \Phi_{1} \otimes \dots  \otimes \Phi_{n}
\end{equation}
and a state $\ket{junk}$ such that, for each $j$, and $m \in \{1, \dots, k_{j}\}$
\begin{eqnarray}
\norm{\Phi(\ket{\psi^{\prime}})  -  \ket{junk}\otimes \ket{\psi}}_{1} & \leq & \delta \\
\norm{\Phi(M^{\prime}_{j,m}\ket{\psi^{\prime}}) - \ket{junk} \otimes M_{j,m}\ket{\psi} }_{1} & \leq& \delta
\end{eqnarray}
where $\delta =  \frac{15 n^{2} + 5n}{2} \sqrt{\epsilon}$ ($\delta = ?$) and $\ket{junk}$ is in the same Hilbert space as $\ket{\psi^{\prime}}$.
\end{definition}

\begin{theorem}
Let a graph $G$ be given with $|V| = n$.  If a compatible physical experiment $\epsilon$-simulates reference experiment 1 (2) then it is $\delta$-equivalent to it with $\delta = \frac{n}{2}\left(5n^{2} + 11n + 4\right) \sqrt{\epsilon}$ ($\delta =  (2n^{3} + 4n^{2} + n) \sqrt{\epsilon} + 13(\frac{1}{2}n^{2} + n) \epsilon^{\frac{1}{4}}$)
\end{theorem}
%

\subsection{Proof for reference experiment 1}
First we note that if $\bra{\psi}M\ket{\psi} \geq 1 - \epsilon$ then
\begin{equation}
\norm{\ket{\psi} - M \ket{\psi}}_{1} \leq \sqrt{2\epsilon}.
\end{equation}
Next, suppose that we have $\norm{\ket{\psi} - M \ket{\psi}}_{1} \leq \alpha$ and $\norm{\ket{\psi} - N \ket{\psi}}_{1} \leq \beta$.  Using the triangle inequality and the fact that $\norm{M}_{\infty} = 1$ we have
\begin{equation}
\norm{\ket{\psi} - MN \ket{\psi}}_{1} \leq \alpha + \beta.
\end{equation}
The remainder of the proof will use these estimations repeatedly, along with the triangle inequality.  We need only count the number of times this happens, which is the same as the number of operators multiplied together.

First, for lemma~\ref{lemma:oddcycleanticommute} let $c$ be the size of the induced cycle.  We multiply $c+1$ operators together.  Thus we conclude that for a vertex $u$ in the induced cycle
\begin{equation}
\norm{(X^{\prime}Z^{\prime})_{u} \ket{\psi^{\prime}} + (Z^{\prime}X^{\prime})_{u} \ket{\psi^{\prime}}}_{1} \leq 2(c+1)\sqrt{\epsilon}.
\end{equation}
Next, for lemma~\ref{lemma:chaining} we multiply four operators, then invoke the anti-commuting property on one of the vertices.  This gives
\begin{equation}
\norm{(X^{\prime}Z^{\prime})_{u} \ket{\psi^{\prime}} + (Z^{\prime}X^{\prime})_{u} \ket{\psi^{\prime}}}_{1} \leq 8\sqrt{\epsilon} + \beta
\end{equation}
where $\beta$ is $\norm{(X^{\prime}Z^{\prime})_{v} \ket{\psi^{\prime}} + (Z^{\prime}X^{\prime})_{v} \ket{\psi^{\prime}}}_{1}$, $v$ being neighbouring vertex.  We may apply lemma~\ref{lemma:chaining} along paths from vertices in the induced cycle in $G$.  Let $l$ be the length (number of edges) of the longest path.  Then for any vertex $u$ we find, at worst,
\begin{equation}
\norm{(X^{\prime}Z^{\prime})_{u} \ket{\psi^{\prime}} + (Z^{\prime}X^{\prime})_{u} \ket{\psi^{\prime}}}_{1} \leq 2(4l + c + 1)\sqrt{\epsilon}.
\end{equation}

Lastly, for lemma~\ref{lemma:exchangexy}, we multiply $|V^{\prime}|$ operators, and apply the anti-commuting relation $|E^{\prime}|$ times.  Thus
\begin{equation}
\norm{(-1)^{|E^{\prime}|} X^{\prime V^{\prime}} \ket{\psi^{\prime}} - Z^{\prime N(V^{\prime})} \ket{\psi^{\prime}}}_{1} \leq 2\left(|V^{\prime}| +  (4l + c + 1)|E^{\prime}|\right)\sqrt{\epsilon}.
\end{equation}

We are now ready to analyze the proof of the main theorem for reference experiment 1.  To arrive at equation~\ref{eq:isometryanticommute} we apply the anti-commutation relation.  This happens once for each 1 appearing in $x$, for each possible $x$, for a total of $n 2^{n-1}$ times.  We may find this by pairing values $x$ and $x \oplus 111\dots1$.  There are $2^{n-1}$ such pairs and each pair contains $n$ 1s all together.  Multiplying by the normalization factor $\frac{1}{2^{n}}$ we find
\begin{equation}
\norm{\Phi(\ket{\psi^{\prime}}) - \frac{1}{2^{n}}\sum_{x} \bigotimes_{vÊ\in V} \left(I + Z^{\prime}_{v}\right)X^{\prime x_{v}}_{v} \ket{\psi^{\prime}}\ket{x}} \leq n(4l + c + 1)\sqrt{\epsilon}.
\end{equation}
For equation~\ref{eq:isometryfinal} we use lemma~\ref{lemma:exchangexy}, once for each possible value of $x$.  Again, the estimate depends on the number of 1s in $x$, summed over all possible $x$s.  As well, it depends on the number of edges in the induced subgraph.  An edge $(u,v)$ will be counted only when $x_{u} = x_{v} = 1$.  This occurs for $1/4$ of all $x$s.  Summed over all possible $x$s and edges, then, the number of times edges are counted is $2^{n-2}|E|$.  Again, we multiply by the normalization factor $\frac{1}{2^{n}}$.  This gives our final estimate:
\begin{equation}
\norm{\Phi(\ket{\psi^{\prime}}) - \left(\frac{1}{2^{n}}\bigotimes_{vÊ\in V} \left(I + Z^{\prime}_{v}\right)\ket{\psi^{\prime}}\right) \sum_{x} (-1)^{e(x)}\ket{x}}_{1} 
\end{equation}
\begin{equation}
\leq \left(n(4l + c + 1) + \left(n+  (4l + c + 1)|E|/2\right)\right)\sqrt{\epsilon}
\end{equation}
\begin{equation}
= \left((4l + c + 1)(n + \frac{|E|}{2}) + n \right) \sqrt{\epsilon}
\end{equation}
where $e(x)$ is the number of edges in the induced subgraph on the set $V_{x} = \{v \in V | x_{v} = 1\}$.  

Note that when calculating $\Phi\left(X^{\prime}_{u} \ket{\psi^{\prime}}\right)$ etc.\ we did not use any more estimations, we simply rearrange when lemma~\ref{lemma:exchangexy} is applied.  Thus the same robustness applies.

As a last estimation, we note that $l$ and $c$ cannot be larger than $n$, and $|E| \leq n^{2}$.  We may thus set $\delta  = \frac{n}{2}\left(5n^{2} + 11n + 4\right) \sqrt{\epsilon}$.

Note that we may make much better estimates if some properties of the graph are known.  For example, if every vertex lies in a triangle and the max degree is $6$, as in the case of a lattice of triangles, we may instead set $\delta = 17n \sqrt{\epsilon}$.

\subsection{Proof for reference experiment 2}
Much of the same analysis may be used for experiment 2.  Indeed, since the only difference in the proofs for the non-robust results is how the anti-commuting property is proved, we may simply replace the estimation for lemma~\ref{lemma:oddcycleanticommute} with that of lemma~\ref{lemma:dxyanticommute}.

We begin, then, with $\epsilon$-simulation and prove a robust version of lemma~\ref{lemma:dxyanticommute}.  First we wish to estimate $\alpha = \norm{D^{\prime}_{u} \ket{\psi} - \frac{X^{\prime}_{u} + Z^{\prime}_{u}}{\sqrt{2}}\ket{\psi}}_{1}$.  Using techniques from the previous section, we have
\begin{eqnarray}
\norm{X^{\prime}_{u} \ket{\psi^{\prime}} - Z^{\prime N_{u}} \ket{\psi}}_{1} & \leq & 2 \sqrt{\epsilon} \\
\norm{Z^{\prime}_{u} \ket{\psi^{\prime}} - X^{\prime}_{v} Z^{\prime N_{v} \setminus u}\ket{\psi^{\prime}}}_{1} & \leq & 2 \sqrt{\epsilon}. 
\end{eqnarray}
These along with the triangle inequality give an upper bound for $\alpha$ of
\begin{equation}
2\sqrt{2\epsilon} + \norm{D^{\prime}_{u} \ket{\psi} - \frac{Z^{\prime N_{u}} + X^{\prime}_{v} Z^{\prime N_{v} \setminus u}}{\sqrt{2}}\ket{\psi}}_{1} 
\end{equation}
Expanding the second term, we get 

\begin{equation}
\sqrt{1 + \norm{\frac{Z^{\prime N_{u}} +X^{\prime}_{v} Z^{\prime N_{v} \setminus u}}{\sqrt{2}} \ket{\psi^{\prime}}}_{1}^{2} - \sqrt{2}\left(\bra{\psi^{\prime}}D^{\prime}_{u}Z^{\prime N_{u}} \ket{\psi^{\prime}} + \bra{\psi^{\prime}}D^{\prime}_{u} X^{\prime}_{v}Z^{\prime N_{v} \setminus u} \ket{\psi^{\prime}}\right)}.
\end{equation}
Since $\norm{Z^{\prime}_{u} \ket{\psi^{\prime}} - X^{\prime}_{v} Z^{\prime N_{v} \setminus u}\ket{\psi^{\prime}}}_{1} \leq 2 \sqrt{\epsilon}$ and $\norm{Z^{\prime N_{u}}\ket{\psi^{\prime}}}_{1} = 1$ we find 
\begin{equation}
\left|\bra{\psi^{\prime}} Z^{\prime N_{u}}Z^{\prime}_{u} \ket{\psi^{\prime}} - \bra{\psi^{\prime}}Z^{\prime N_{u}}X^{\prime}_{v} Z^{\prime N_{v} \setminus u}\ket{\psi^{\prime}}\right| \leq 2 \sqrt{\epsilon}.
\end{equation}
By hypothesis, $\left|\bra{\psi^{\prime}} Z^{\prime N_{u}}Z^{\prime}_{u} \ket{\psi^{\prime}} \right| \leq \epsilon$, so $\left|\bra{\psi^{\prime}}Z^{\prime N_{u}}X^{\prime}_{v} Z^{\prime N_{v} \setminus u}\ket{\psi^{\prime}}\right| \leq 2 \sqrt{\epsilon} + \epsilon$.

Meanwhile 
$\beta^{2} = \norm{\frac{Z^{\prime N_{u}} +X^{\prime}_{v} Z^{\prime N_{v} \setminus u}}{\sqrt{2}} \ket{\psi^{\prime}}}_{1}^{2} = 1 + \text{Re} \bra{\psi^{\prime}} Z^{\prime N_{u}} X^{\prime}_{v} Z^{\prime N_{v} \setminus u} \ket{\psi^{\prime}}$, so $|1- \beta^{2}| \leq 2\sqrt{\epsilon} + \epsilon$.

Finally, by hypothesis $\left|\bra{\psi^{\prime}}D^{\prime}_{u}Z^{\prime N_{u}} \ket{\psi^{\prime}} + \bra{\psi^{\prime}}D^{\prime}_{u} X^{\prime}_{v}Z^{\prime N_{v} \setminus u} \ket{\psi^{\prime}} - \sqrt{2}\right| \leq 2\epsilon$.  Combining these facts we find 
$\alpha \leq 2\sqrt{2 \epsilon} + \sqrt{2 \sqrt{\epsilon} + (1 + 2\sqrt{2}) \epsilon}$.

Now we wish to estimate
\begin{equation}
\norm{(D_{u}^{\prime})^{2}\ket{\psi^{\prime}} - \frac{\left(X^{\prime}_{u} + Z^{\prime}_{u}\right)^{2}}{2} \ket{\psi^{\prime}}}_{1}
\end{equation}
By the fact $\norm{D^{\prime}_{u}}_{\infty} = 1$ we have $\norm{(D_{u}^{\prime})^{2}\ket{\psi^{\prime}} - D^{\prime}_{u}\frac{X^{\prime}_{u} +Z^{\prime}_{u}}{\sqrt{2}} \ket{\psi^{\prime}}}_{1} \leq \alpha$.  Similarly, since $\norm{X^{\prime}_{u} + Z^{\prime}_{u}}_{\infty} \leq 2$ we find $\norm{D_{u}^{\prime}\frac{X^{\prime}_{u} Z^{\prime}_{u}}{\sqrt{2}}\ket{\psi^{\prime}} - \frac{\left(X^{\prime}_{u} + Z^{\prime}_{u}\right)^{2}}{2} \ket{\psi^{\prime}}}_{1} \leq \sqrt{2} \alpha$.  Using these facts, the triangle inequality, and $(D^{\prime}_{u})^{2} = I$, we obtain

\begin{equation*}
2\norm{\ket{\psi^{\prime}} - \frac{\left(X^{\prime}_{u} + Z^{\prime}_{u}\right)^{2}}{2} \ket{\psi^{\prime}}}_{1} = \norm{X^{\prime}_{u} Z^{\prime}_{u}\ket{\psi^{\prime}} + Z^{\prime}_{u} X^{\prime}_{u} \ket{\psi^{\prime}}}_{1} 
\end{equation*}
\begin{equation}
\leq 2(1 + \sqrt{2})\left(2\sqrt{2 \epsilon} + \sqrt{2 \sqrt{\epsilon} + (1 + 2\sqrt{2}) \epsilon}\right) \leq 26 \epsilon^{\frac{1}{4}}
\end{equation}
with the last inequality valid for $\epsilon \leq 1$.

Using this estimate, and working through the estimations as in the previous section, we find that we may set 
\begin{equation}
\delta = (2l (2n + |E|) + n) \sqrt{\epsilon} + 13(n + \frac{1}{2}|E|) \epsilon^{\frac{1}{4}}.
\end{equation}
For a simpler expression, we may use $l \leq n$ and $|E| \leq n^{2}$, obtaining
\begin{equation}
\delta = (2n^{3} + 4n^{2} + n) \sqrt{\epsilon} + 13(\frac{1}{2}n^{2} + n) \epsilon^{\frac{1}{4}}.
\end{equation}

Again, we may find a better estimate with more information about the graph.  For cluster states, which have a square lattice graph, we have $|E| \leq 4n$.  We may also perform $D_{u}$ measurements on all vertices and set $l=0$.  In this case we may set $\delta = n \sqrt{\epsilon} + 39n\epsilon^{\frac{1}{4}}$.

\section{Discussion}

\subsection{Estimating expected values}
The main results concern expected values, rather than experimental outcomes.  So in order to make use of these results in any practical implementation we must estimate the expected values using data collected from experimental outcomes.  The obvious approach of sampling the devices many times and applying a Chernoff bound is problematic.  In particular, we do not wish to assume that separate uses of a device are independent and identically distributed since these assumptions would be untestable and likely false in many practical experiments.

One approach to this problem is that used by Pironio et al. in \cite{Pironio:2010:Random-numbers-}.  There the authors construct a martingale, which is a sequence of random variables with certain properties.  In particular, the random variables need not be independent.  This allows them to use Azuma's inequality, which gives good bounds for martingales on how far away a sample may lie from the expected value without relying on independence assumptions.  A similar approach is viable here and a preliminary analysis suggests that good bounds are achievable.

\subsection{Graph state computation}
Graph states are particularly interesting for their role in measurement based quantum computation (MBQC, \cite{Raussendorf:2001:A-One-Way-Quant}).  In this paradigm a graph state is measured, vertex by vertex, in particular bases.  Each measurement may be interpreted as performing a unitary on a logical qubit.  The composition of these unitaries performs a logical circuit on the logical qubits.

A natural question to ask is whether a self-tested graph state could be used for MBQC to allow for self-tested computation.  Unfortunately MBQC depends on measurements in the $X$-$Y$ plane and the measurements tested here are all in the $X$-$Z$ plane.  However, the techniques used in \cite{McKague:2010:Generalized-sel} could easily be adapted to allow testing of $X$-$Y$ plane measurements which would then allow self-tested MBQC.  In fact, in the exact case the techniques used in \cite{McKague:2010:Generalized-sel} can be used with minimal changes.  A preliminary analysis of robustness suggests that the errors scale similarly to that of lemma~\ref{lemma:dxyanticommute} here.

\bibliography{Global_Bibliography}

\appendix
\section{Classical hidden variable model for bipartite graph states with $X$ and $Z$ measurements}\label{sec:bipartiteneedsd}

Let $G$ be a bipartite graph and $\ket{\psi}$ the corresponding graph state.  We give a local hidden variable model that is consistent will all measurements which are tensor products of $X$ and $Z$ on this state.

We construct a local hidden variable model by randomly choosing a value $\pm 1$ for $Z^{\prime}_{v}$ for each $v$ in the graph.  We then set $X^{\prime}_{v}$ to be
\begin{equation}
X^{\prime}_{v} = \prod_{u \in N_{v}} Z^{\prime}_{u}.
\end{equation}

Now we show that this is consistent with all possible tensor product $X$ and $Z$ measurements on $\ket{\psi}$.  Let $M = X^{S}Z^{T}$, $S \cap T = \emptyset$ be such a measurement.  First, suppose that $\pm M$ can be written as a product of stabilizers of $\ket{\psi}$.  Using lemma~\ref{lemma:evendegreeoddedges} we have
\begin{equation}
M = X^{S} Z^{N(S)} =   (-1)^{|E(S)|} \prod_{x \in S} S_{v}.
\end{equation}
Note that, by assumption, $M$ has only $X$ and $Z$ factors, so each $v \in S$ must have an even number of neighbours in $S$.  Then the induced subgraph on $S$ is Eulerian and we can partition the edges of the subgraph into cycles with no common edges (see Diestel \cite{Diestel:2010:Graph-Theory} for a proof).  Suppose that $|E(S)|$ is odd.  Then there must be at least one odd cycle in this partition and then $S$ has an odd cycle and so does $G$.  Since $G$ is bipartite this must not be the case and in fact $|E(S)|$ is even.  Hence $M = \prod_{x \in S} S_{v}$ and $\bra{\psi}M\ket{\psi} = 1$.  By construction $M^{\prime} = X^{\prime S} Z^{\prime N(S)} = \prod_{v \in S} X^{\prime}_{v} Z^{\prime N_{v}} = 1$ and the expected value of $M^{\prime}$ matches that of $M$.

Now suppose that $M$ is not a product of stabilizers of $\ket{\psi}$.  Then $M$ must anti-commute with at least one stabilizer and hence $\bra{\psi}M\ket{\psi} =0$.  Meanwhile, by construction
\begin{equation}
M^{\prime} = X^{\prime S} Z^{\prime T} = Z^{\prime N(S)}Z^{\prime T.}.
\end{equation}
If $N(S) = T$ then $M$ is in fact a product of stabilizers.  This is not the case, so there is at least one $Z^{\prime}_{v}$ in the above equation which is not cancelled.  Since all the $Z^{\prime}_{v}$s are chosen randomly, the product of the $Z^{\prime}_{v}$s not cancelled will also be uniformly random.  Thus the expected value of $M^{\prime}$ is 0.

\end{document}